\documentclass[letterpaper,conference]{ieeeconf}
\usepackage[english]{babel}
\IEEEoverridecommandlockouts

\usepackage{microtype}
\usepackage{setspace}
\usepackage{blindtext}
\usepackage{siunitx}
\usepackage{titlecaps}
\Addlcwords{or the with without no if of for and to -off off in by via}

\usepackage{cite}

\usepackage{graphicx}
\usepackage{subfloat}
\usepackage{fancyhdr} 
\usepackage{color}
\usepackage{epsfig}
\graphicspath{{Images/}{Images/Comparison/}{Images/Error/}{Images/Graph/}}
\usepackage{pgfplots}
\usepackage[justification=raggedright,singlelinecheck=false,size=footnotesize]{caption}
\usepackage[justification=raggedright,singlelinecheck=false,size=footnotesize]{subcaption}

\usepackage{array}
\usepackage{stfloats}

\usepackage{amsthm,amsmath,amssymb,amsfonts}
\usepackage{dsfont}
\usepackage{relsize}
\usepackage{nicefrac}
\usepackage{bbm}
\usepackage{mathtools}
\usepackage[mathscr]{eucal}

\usepackage[ruled]{algorithm}
\usepackage{algpseudocode}
\let\labelindent\relax
\usepackage{enumitem}

\usepackage{float}
\usepackage[absolute]{textpos}

\usepackage{url}
\usepackage[bookmarks=true]{hyperref}
\usepackage[capitalize]{cleveref}

\newcommand\orcidicon[1]{\href{https://orcid.org/#1}{\includegraphics[scale=0.05]{orcid}}}

\newcommand{\node}{agent\xspace}
\newcommand{\nodes}{agents\xspace}

\newcommand{\Real}[1]{ { {\mathbb R}^{#1} } }

\DeclareMathOperator*{\argmin}{arg\,min}

\newcommand{\tr}[1]{\mbox{tr}\left(#1\right)}

\newcommand{\E}[1]{\mathbb{E}\left[#1\right]}
\newcommand{\Enorm}[1]{\mathbb{E}\left[\left\|#1\right\|^2\right]}
\newcommand{\inv}{^{-1}}

\newcommand{\one}{\mathds{1}}

\newcommand{\ereg}{e_\regSet}

\theoremstyle{plain}

\newtheorem{cor}{Corollary}
\newtheorem{prop}{Proposition}
\newtheorem{lemma}{Lemma}
\theoremstyle{definition}

\newtheorem{ass}{Assumption}
\theoremstyle{remark}
\newtheorem{rem}{Remark}

\newcommand{\sensSet}{\mathcal{V}}
\newcommand{\malSet}{\mathcal{M}}
\newcommand{\regSet}{\mathcal{R}}
\newcommand{\xnode}[2]{x_{#1}(#2)}

\newcommand{\xreg}{x_\regSet}

\newcommand{\neigh}[1]{\mathcal{N}_{#1}}
\newcommand{\priornode}[1]{\theta_{#1}}
\newcommand{\priorall}{\theta}
\newcommand{\priornodeout}[1]{\tilde{\theta}_{#1}}
\newcommand{\priorallout}{\tilde{\theta}}
\newcommand{\priorreg}{\theta_\regSet}

\newcommand{\thbar}{\bar{\theta}}
\newcommand{\thbarout}{\bar{\priorallout}}
\newcommand{\thbarreg}{\bar{\theta}_\regSet}
\newcommand{\thbarmal}{\bar{\tilde{\theta}}_\malSet}

\newcommand{\Var}{\tilde{\Sigma}}

\newcommand{\Creg}{C_R}
\newcommand{\selreg}{S_R}
\newcommand{\Emat}{\selreg L - \Creg\selreg}

\newcommand{\lami}{\lambda_i}
\newcommand{\lam}{\lambda}

\newcommand{\sumall}[1]{\sum_{#1\in\sensSet}}
\newcommand{\sumneigh}[2]{\sum_{#1\in\neigh{#2}}}
\newcommand{\sumreg}[1]{\sum_{#1\in\regSet}}

\algdef{SE}[DOWHILE]{Do}{doWhile}{\algorithmicrepeat}[1]{\algorithmicuntil\ #1}%

\renewcommand{\algorithmiccomment}[1]{\bgroup\hfill//~#1\egroup}

\addto\captionsenglish{}
\addto\captionsenglish{}

\newcommand{\blue}[1]{{\color{blue}#1}}

\newcommand{\revision}[1]{#1}

\newcommand{\linkToPdf}[1]{\href{#1}{\blue{(pdf)}}}
\newcommand{\linkToPpt}[1]{\href{#1}{\blue{(ppt)}}}
\newcommand{\linkToCode}[1]{\href{#1}{\blue{(code)}}}
\newcommand{\linkToWeb}[1]{\href{#1}{\blue{(web)}}}
\newcommand{\linkToVideo}[1]{\href{#1}{\blue{(video)}}}
\newcommand{\linkToMedia}[1]{\href{#1}{\blue{(media)}}}
\newcommand{\award}[1]{\xspace} 
\newcommand{\eg}{\emph{e.g.,}\xspace}
\newcommand{\ie}{\emph{i.e.,}\xspace}
\addto\extrasenglish{}
\addto\extrasenglish{}
\addto\extrasenglish{}

\title{\LARGE\bfseries\titlecap{competition-based resilience in distributed quadratic optimization}}

\author{Luca~Ballotta$ ^1 $, Giacomo Como$ ^2 $, Jeff S. Shamma$^3$ and Luca~Schenato$ ^1 $%
	\thanks{This work has been partially funded by 
		the Italian Ministry of Education, University and Research (MIUR) through the PRIN project no. 2017NS9FEY entitled ``Realtime Control of 5G Wireless Networks: Taming the Complexity of Future Transmission and Computation Challenges'' and through the initiative "Departments of Excellence" (Law 232/2016). The views and opinions expressed in this work are those of the authors and do not necessarily reflect those of the funding institutions.}%
	\thanks{$ ^1 $Department of Information Engineering, University of Padova, 35131 Padova, Italy
		{\tt\small \{ballotta, schenato\}@dei.unipd.it}}%
	\thanks{$ ^2 $Department of Mathematical Sciences, Politecnico di Torino, Corso Duca degli Abruzzi 24, 10129, Torino, Italy 
		{\ttfamily\small giacomo.como@polito.it}}%
	\thanks{$ ^3 $Industrial and Enterprise Systems Engineering (ISE), University of Illinois Urbana-Champaign, Illinois, USA  
		{\ttfamily\small jshamma@illinois.edu}}
}

\begin{document}
	\bstctlcite{MyBSTcontrol}
	\begin{textblock}{20}(-2,0.05)
		\footnotesize
		\centering
		\setstretch{1}
		This paper has been accepted for the 61th IEEE Conference on Decision and Control in Cancun, December 6-9 2022.\\
		Please cite the paper as: L. Ballotta, G. Como, J. S. Shamma, and L. Schenato,\\
		“\titlecap{competition-based resilience in distributed quadratic optimization}”,\\
		IEEE Conference on Decision and Control, 2022.
	\end{textblock}
	\begin{textblock}{10}(3,15)
		\footnotesize
		\centering
		\setstretch{1}
		\textcopyright 2022 IEEE.  
		Personal use of this material is permitted.  
		Permission from IEEE must be obtained for all other uses, in any current or future media, including reprinting/republishing this material for advertising or promotional purposes, creating new collective works, for resale or redistribution to servers or lists, or reuse of any copyrighted component of this work in other works.
	\end{textblock}
	\maketitle
	

\begin{abstract}

	This paper proposes a novel approach to resilient distributed optimization with quadratic costs
	in a networked control system (\eg wireless sensor network, power grid, robotic team)
	prone to external attacks (\eg hacking, power outage) that cause \nodes to misbehave.
	Departing from classical filtering strategies proposed in literature,
	we draw inspiration from a game-theoretic formulation of the consensus problem
	and argue that adding competition to the mix can enhance resilience
	in the presence of malicious \nodes.
	Our intuition is corroborated by analytical and numerical results showing that
	i) our strategy highlights the presence of a nontrivial trade-off between blind collaboration and full competition, and
	ii) such competition-based approach can outperform state-of-the-art algorithms based on Mean Subsequence Reduced.
\end{abstract}

\section{Introduction}\label{sec:intro}

\textquotedblleft With great power comes great responsibility",
and Networked Control Systems have great power indeed.
From power grids regulating energy consumption,
to large-scale sensor networks able to monitor vast environments,
to fleets of autonomous vehicles for intelligent transportation,
everyday life depends more and more on
control of interacting devices. 

While this brings numerous benefits, 
a major drawback is that malicious \node can locally 
intrude from any point in the system, 
and cause major damage at a global scale.
Recently,
Department of Energy secretary
stated that enemies of the United States can shut down the U.S. power grid,
and it is known that hacking groups around the world have high technological sophistication~\cite{USpowergridattack}.
Cyberattacks hit Italian health care IT infrastructures during the COVID pandemic,
disrupting services for weeks~\cite{healthcareItahacking}.
Another concern is represented by failure cascades 
generating from single nodes. 
Damages from cascading failures have notable examples across many domains,
from electric blackouts over large areas,
to denial of service or malfunctioning of web applications.

Such problems have been tackled in literature for several years,
with focus on, 
\eg power outage~\cite{6733531},
cascading failure~\cite{6915846,8070359},
or denial of service~\cite{8629941}.
Related literature in control theory has mostly focused on robustness
of distributed optimization algorithms,
with particular attention to so-called resilient consensus.
For example, average consensus represents a fundamental tool
in many applications, ranging from robot coordination to management of power grids.
However, the classical consensus algorithm is fragile 
in the presence of misbehaving nodes,
which may easily deceive the rest of the network.
To contrast such a problem, the most popular approach is based on the
filtering technique named \textquotedblleft Mean Subsequence Reduced" (MSR),
where each node filters out the largest and smallest incoming values.
In particular,
seminal work~\cite{6481629} improved classical MSR
introducing a weighted version (W-MSR) and the notion of $ r $\textit{-robustness} of a graph,
that was demonstrated to be a suitable measure enabling tight conditions for MSR-based resilient consensus.
Among the variants of W-MSR,
~\cite{DIBAJI201523} studies resilient control for double integrators,
~\cite{WANG20203409} tackles mobile malicious agents,
and~\cite{SHANG2020109288} targets nonlinear systems with state constraints.

Despite the success MSR-based approaches have known in literature,
and also the effectiveness showed in many applications domains,
a minimal level of network robustness is required for
theoretical guarantees of reaching resilient consensus.
In fact, 
little can be said about the system behavior when robustness conditions do not hold.
While sometimes algorithms still work in practice,
the lack of theoretical guarantees may be undesired in certain applications,
where a more conservative but safer approach may be preferred instead.
Moreover, while in certain applications the \nodes may agree upon any common value,
in other cases \emph{average consensus} plays a crucial role.
Thus, we depart from the classical MSR-based strategies with the aim
of designing distributed algorithms that can ensure certain levels of \revision{resilience},
and whose performance can be characterized, beyond the
hard constraint of achieving consensus by all means,
addressing the distance from the optimal solution as performance metric in a distributed optimization task.

To this aim, we draw inspiration by game theory.
Distributed cooperative control and game theoretic approaches, despite their apparent contrast,
are bounded from several perspectives
which have been largely explored in literature~\cite{MARDEN2015861,PROSKURNIKOV201765}.
Our starting point is~\cite{4814554},
where the authors discuss the connection between consensus and potential games.
Stepping forward, we propose a novel approach to resilient consensus
based on the celebrated Friedkin-Johnsen (FJ) dynamics~\cite{FJdynamics},
where an agent in the network can trade collaboration with neighbors for
selfishness, which we interpret as competition against other \nodes. 
Such a mixed approach allows us to explore the performance trade-off that arises
from different levels of collaboration in the network, 
which turns out to be crucial in the presence of attacks.
In fact, we observe that \textbf{the global network cost is minimized by a hybrid strategy
	whereby \nodes trust each other only partially}.

Towards this goal,
we first introduce and motivate average consensus for distributed quadratic optimization tasks in~\autoref{sec:setup}.
Then, we address the presence of outliers and misbehaving \nodes in the network,
and propose a competition-based strategy to enhance resilience to such adversaries in~\autoref{sec:resilient-strategy}.
In particular, we characterize the performance of our approach analytically (\autoref{sec:malicious})
and substantiate theoretical findings with numerical tests on the cost function (\autoref{sec:numerical-tests-FJ-err}).
To further reinforce the validity of our approach on actual problem instances,
we perform numerical simulations on large-scale networked systems
and compare to W-MSR,
showing that competition-based strategies can provide superior performance
if robustness requirements of MSR-based algorithms are not satisfied (\autoref{sec:literature-comparison}).
We conclude this paper by addressing some open questions and avenues for future research in~\autoref{sec:future-research}.

\section{\titlecap{setup}}\label{sec:setup}

We consider a networked control system composed of $ N $ \nodes,
collected in the set $ \sensSet = \{1,\dots,N\} $.
The state of \node $ i\in\sensSet $ is denoted as $ x_i\in\Real{} $,
and all states are stacked in the column vector $ x\in\Real{N} $.
Each \node $ i $ has a \emph{prior} $ \priornode{i} $,
and needs 
to minimize the mismatches with all priors,
\begin{equation}\label{eq:agent-cost-orig}
	f_i(x_i) \doteq \sum_{j\in\sensSet} \left(x_i - \priornode{j}\right)^2.
\end{equation}

\begin{ass}[Prior distribution]\label{ass:prior-distribution}
	The priors are distributed as i.i.d. random variables with zero mean and unit variance.
\end{ass}

Straightforward algebraic manipulations show that the optimal values of $ x_i $, $ i\in\sensSet $,
minimizing~\eqref{eq:agent-cost-orig} correspond to the average of the priors of all agents. 
Indeed,
\begin{equation}\label{eq:network-cost-orig}
	\begin{aligned}
		f(x) &\doteq \dfrac{1}{N}\sumall{i}f_i(x_i) = \dfrac{1}{N}\sumall{i}\sumall{j} \left(x_i - \priornode{j}\right)^2\\
			 &= \dfrac{1}{N}\sumall{i}\sumall{j}\left(x_i^2+\priornode{j}^2-2x_i\priornode{j}\right)\\
			 &= \dfrac{1}{N}\sumall{i}\left(Nx_i^2 + \sumall{j}\priornode{j}^2 - 2x_i\sumall{j}\priornode{j}\right)\\
			 &= \sumall{i}\left(x_i^2 - 2x_i\thbar + \dfrac{1}{N}\sumall{j}\theta_j^2\right)\\
			 &= \sumall{i}\left(x_i^2 - 2x_i\thbar +\thbar^2 \right) - N\thbar^2 + \sumall{j}\priornode{j}^2\\
			 &= \sumall{i}\left(x_i - \thbar\right)^2 + \sumall{i}\left(\theta_i^2 - \thbar^2\right)\\
			 &= \sumall{i}\left(x_i - \thbar\right)^2 + \mathrm{const},
	\end{aligned}
\end{equation}
where $ \thbar $ is the average of all priors.
In light of~\eqref{eq:network-cost-orig}, all \nodes need to reach \emph{average consensus}
to solve the optimization task. 
Hence, we are interested in minimizing the \emph{expected average consensus error} (or simply consensus error), defined as
\begin{equation}\label{eq:cons-error}
	e \doteq \E{\sumall{i}\left(x_i - \thbar\right)^2} = \E{\left\|x-C\priorall\right\|^2},
\end{equation}
where the expectation is taken with respect to the prior distribution,
$ \priorall\in\Real{N} $ stacks all priors, 
$ C \doteq \frac{1}{N}\one_N\one_N^{\top} $,
and $ \one_N \in \Real{N}$ is the vector of all ones.
According to the networked structure,
\nodes can communicate with neighbors to decrease local costs and reach consensus dynamically.

\begin{ass}[Communication network and weights]\label{ass:W-doubly-stochastic}
	The network topology is defined	by an irreducible symmetric doubly-stochastic matrix $ W $.
\end{ass}


\begin{ass}[Local optimization]\label{ass:no-self-loops}
	In view of the optimization tasks, 
	we require no self-loops, \ie $ W_{ii} = 0 \, \forall\,i\in\sensSet $.
\end{ass}

It is well known that, to reach average consensus under nominal conditions and~\cref{ass:W-doubly-stochastic},
the \nodes can update their states according to the so-called consensus dynamics,
\begin{equation}\label{eq:consensus-protocol}
	\xnode{i}{k+1} = \sumneigh{j}{i}W_{ij}\xnode{j}{k}, \qquad \xnode{i}{0} = \priornode{i},
\end{equation}
where $ \xnode{i}{k} $ is the state of \node $ i $ at time $ k\ge0 $
and $ \neigh{i}\subseteq\sensSet $ is its neighborhood in the topology $ W $.

\section{\titlecap{strategies for resilient consensus}}\label{sec:resilient-strategy}

\subsection{\titlecap{average consensus in the presence of outliers}}\label{sec:outliers}

We now assume that some \nodes are \emph{outliers},
\ie their priors are corrupted by additive noise.
The \nodes not affected by noise are called \emph{regular}.
We denote the subsets of outliers and regular \nodes as $ \malSet $ and $ \regSet $, respectively,
$ M \doteq |\malSet| $ and $ R \doteq |\regSet| $.
Formally speaking, the \emph{actual prior} of corrupted \node $ m\in\malSet $ is
$ \priornodeout{m} = \priornode{m} + v_m $,
where $ v_m $ is a zero-mean random variable with variance $ d $.
We assume that noises are uncorrelated with each other and with priors of all \nodes.

The average of priors in the presence of outliers differs from its nominal value $ \thbar $,
which prevents reaching average consensus via protocol~\eqref{eq:consensus-protocol}.
Indeed, the new average is
\begin{equation}\label{eq:prior-average-outliers}
	\begin{aligned}
		\thbarout = \dfrac{1}{N}\left(\sum_{i\in\regSet}\priornode{i} + \sum_{m\in\malSet}\priornodeout{m}\right) 
		= \thbar + \dfrac{1}{N}\sum_{m\in\malSet}v_m,
	\end{aligned}
\end{equation}
%
and protocol~\eqref{eq:consensus-protocol} drives \nodes
to the consensus error
\begin{equation}\label{eq:cons-error-outliers}
	e^C = \E{\sumall{i}\left(\tilde{\thbar} - \thbar\right)^2} = \dfrac{dM}{N}.
\end{equation}
Aiming to improve the standard consensus dynamics,
we draw inspiration from~\cite{4814554},
where the authors show that~\eqref{eq:consensus-protocol} can be interpreted,
in a game theoretic framework,
as the optimal action in a quadratic game where \node $ i $ repeatedly seeks to maximize the following utility function,
\begin{equation}\label{eq:utility-consensus}
	u_i(x_i) = - \sum_{j\in\neigh{i}} W_{ij}\left(x_i - x_j\right)^2.
\end{equation}
Note that~\eqref{eq:utility-consensus} forces each \node
to get as close as possible to its neighbors.
We modify the above utility as follows,
\begin{equation}\label{eq:utility-FJ}
	u_i^{FJ}(x_i) = -\lami\left(x_i-\priornode{i}\right)^2 - (1-\lami)\sumneigh{j}{i}W_{ij}\left(x_i-x_j\right)^2,
\end{equation}
with $ \lam_i\in[0,1] $.
Now, the best action to maximize~\eqref{eq:utility-FJ} at time $ k $ 
coincides with the Friedrick-Johnsen (FJ) dynamics~\cite{FJdynamics},
\begin{equation}\label{eq:FJ-dynamics}
	\xnode{i}{k+1} = \lami\priornode{i} + (1-\lami)\sumneigh{j}{i}W_{ij}\xnode{j}{k}.
\end{equation}
In words,~\eqref{eq:utility-FJ}--\eqref{eq:FJ-dynamics} make
the \nodes anchor to their priors
proportionally to parameters $ \lami $'s.
This allows regular \nodes that communicate with outliers
to partially retain their uncorrupted prior information,
whereas totally trusting their neighbors could drive them
very distant from the nominal average consensus.
\revision{Specifically,
\emph{full collaboration}~\eqref{eq:utility-consensus} is the special case of~\eqref{eq:utility-FJ}
with $ \lami \equiv 0 $, 
while $ \lami \equiv 1 $ refers to \emph{full competition}.
For the sake of simplicity, we set $ \lam_i \equiv \lam $ in the following
and leave a more detailed analysis for future work.}

It turns out that 
standard consensus~\eqref{eq:consensus-protocol} still yields the best performance
if the outliers are not malicious, \ie they follow the prescribed protocol.

\begin{prop}[Consensus protocol \emph{vs.} FJ dynamics with outliers]\label{prop:cons-err-out}
	In the presence of outliers, 
	consensus protocol~\eqref{eq:consensus-protocol} yields smaller error than
	the FJ dynamics~\eqref{eq:FJ-dynamics} for any $ \lam > 0 $.
\end{prop}
\begin{proof}
	It is known that, if $ \lam W $ is Schur,
	the steady state reached by the FJ dynamics is~\cite{7577815}
	\begin{equation}\label{eq:FJ-dynamics-steady-state}
		x = L\priorallout, \qquad L \doteq \left(I-(1-\lam)W\right)^{-1}\lam,
	\end{equation}
	with $ \priorallout_i = \priornode{i}\,\forall \,i\in\regSet $.
	Notice that, being $ W $ symmetric and doubly stochastic, it holds $ L = L^{\top} $ and $ LC = C $.
	Let the outlier noises be stacked in the vector $ v\in\Real{N} $
	with $ v_i = 0 \,\forall \,i\in\regSet $.
	The consensus error under the FJ dynamics is then
	\begin{equation}\label{eq:cons-error-FJ-outliers}
		\begin{aligned}
			e^{FJ}_\lambda &= \E{\left\|L\priorallout - C\priorall\right\|^2} = \E{\left\|(L-C)\tilde{\priorall} + Cv\right\|^2}\\
			  &\overset{(i)}{=} \E{\left\|(L-C)\priorallout\right\|^2} + \E{\left\|Cv\right\|^2}\\
			  &\overset{(ii)}{=} \E{\left\|(L-C)\priorallout\right\|^2} + e^C \ge e^C,
		\end{aligned}
	\end{equation}
	where symmetry of $ L $ and $ C $ is used in $ (i) $ together with
	the facts $ C^2 = C $ and $ LC = C $ which
	annihilate the inner product, 
	and $ \E{\left\|Cv\right\|^2} = e^C $ is used in $ (ii) $. 
\end{proof}
\begin{rem}[Limit behavior of FJ dynamics]\label{rem:L-limit}
	As noted above, the FJ dynamics~\eqref{eq:FJ-dynamics} tends to protocol~\eqref{eq:consensus-protocol} as $ \lambda\searrow0 $.
	Intuitively, the steady state~\eqref{eq:FJ-dynamics-steady-state} should also tend to average consensus.
	This is indeed the case, as $\lim_{\lam\rightarrow0^+}L=\lim_{k\rightarrow+\infty}W^k$~\cite{7577815}.
\end{rem}

\subsection{\titlecap{average consensus in the presence of malicious agents}}\label{sec:malicious}

We now assume that outliers are also malicious, \ie
they do not obey the prescribed protocol.

\begin{ass}[Malicious \node dynamics]\label{ass:mal-node-dynamics}
	The state of malicious agents is constant,
	\ie $ \xnode{m}{k} \equiv \priornodeout{m} \ \forall \,m\in\malSet $.
	Accordingly, the rows in the matrix $ W $ corresponding to malicious \nodes
	have all off-diagonal elements equal to $ 0 $,
	while the elements on the main diagonal are set to $ 1 $.
\end{ass}

\begin{rem}[Malicious \nodes disrupt optimization]
	\cref{ass:mal-node-dynamics} is consistent with the resilient consensus literature,
	where often algorithms are tested against constant or drifting malicious \nodes
	that keep pulling their neighbors far off nominal average consensus~\cite{WANG20203409,9468419}.
	On the other hand,
	attackers may behave cleverly
	to not be detected,
	which needs to be tamed by suitable resilient strategies.
	However,
	this scenario goes beyond the scope of this paper.
\end{rem}

In this scenario, 
malicious \nodes in the set $ \malSet $ do not collaborate to minimize cost~\eqref{eq:network-cost-orig}.
Hence, we restrict the network optimization problem to the set of regular \nodes, 
redefining their local costs as (cf.~\eqref{eq:agent-cost-orig})
\begin{equation}\label{eq:network-cost-regular}
	f_i\left(x_i\right) = \dfrac{1}{R}\sumreg{j} \left(x_i - \priornode{j}\right)^2, \qquad i \in\regSet,
\end{equation}
and 
address the consensus error across regular \nodes (cf.~\eqref{eq:cons-error}),
\begin{equation}\label{eq:cons-error-regular}
	e_\regSet \doteq \E{\sumreg{i}\left(x_i - \thbarreg\right)^2} = \E{\left\|\xreg-\one_R\thbarreg\right\|^2},
\end{equation}
where $ \xreg \in \Real{R} $ and $ \priorreg \in \Real{R} $ stack
respectively states and priors of regular \nodes.

\subsubsection{Trivial FJ Dynamics}\label{eq:FJ1-better-consensus}

A first remarkable result is that simply letting regular \nodes
not update their states may be sufficient to outperform the standard consensus dynamics
in the presence of malicious \nodes.

\begin{prop}[Consensus protocol \emph{vs.} FJ dynamics with malicious \nodes]\label{prop:cons-err-mal}
	In the presence of malicious \nodes,
	the FJ dynamics~\eqref{eq:FJ-dynamics} with $ \lambda = 1 $
	yields smaller error than consensus protocol~\eqref{eq:consensus-protocol}
	if the noise intensity $ d $ is large enough.
\end{prop}
\begin{proof}
	We first compute the error induced by standard consensus~\eqref{eq:consensus-protocol}.
	In virtue of~\cref{ass:mal-node-dynamics},
	the steady-state consensus value is the average of malicious \node' priors,
	$ \thbarmal \doteq \frac{1}{M}\sum_{m\in\malSet}\priornodeout{m} $.
	The consensus error becomes
	\begingroup
	\allowdisplaybreaks
	\begin{align*}
			e^C &= \E{\sumreg{i}\left(\frac{1}{M}\sum_{m\in\malSet}\priornodeout{m} - \frac{1}{R}\sum_{i\in\regSet}\priornode{i}\right)^2}\\
				&\overset{(i)}{=} R\E{\left(\frac{1}{M}\sum_{m\in\malSet}\priornodeout{m}\right)^2} + R\E{\left(\frac{1}{R}\sum_{i\in\regSet}\priornode{i}\right)^2}\\
				&\overset{(ii)}{=} \dfrac{R}{M^2}\E{\left(\sum_{m\in\malSet}\priornode{m}\right)^2}+\dfrac{R}{M^2}\E{\left(\sum_{m\in\malSet}v_m\right)^2} + 1\\
				&\overset{(iii)}{=} \dfrac{R}{M} + \dfrac{Rd}{M} + 1.
	\end{align*}
	\endgroup
	On the other hand, the FJ dynamics with $ \lam = 1 $ simply freezes all regular \nodes' priors,
	yielding consensus error
	\begin{equation}\label{eq:FJ-err-mal-lamda-1}
		\begin{aligned}
			e_1^{FJ} &= \Enorm{\priorreg - \Creg\priorreg} = \Enorm{(I_R - \Creg)\priorreg}\\
					 &= \tr{\E{\priorreg\priorreg^\top}(I-\Creg)}= R-1,
		\end{aligned}
	\end{equation}
	\setlength\marginparwidth{30pt}
	It follows that, if 
		\begin{equation}\label{eq:cons-FJ-comp}
			d > M\left(1 - \dfrac{2}{R}\right) - 1,
	\end{equation}
	then $ e^C > e_1^{FJ} $ .
\end{proof}
\cref{prop:cons-err-mal} states that, if outlier noises $ v_m $
are sufficiently intense, 
\revision{full competition where \nodes anchor to their priors by setting $ \lam = 1 $ in~\eqref{eq:FJ-dynamics}
provides better performance than full collaboration with $ \lam = 0 $,
\ie standard consensus~\eqref{eq:consensus-protocol}.}

\subsubsection{Optimizing FJ Dynamics}\label{sec:opt-lam-less-1}


We are now interested in choosing $ \lam $ to reduce the consensus error.
In particular, we aim to characterize the optimal parameter,
which we denote as $ \lam^* \doteq \argmin_{\lam}\ereg $.
Note that continuity of $ \ereg $ (with its continuous extension at $ \lam = 0 $, see~\cref{rem:L-limit})
and Weierstrass theorem ensure that such $ \lam^* $ always exists.
In the following results, we use the covariance matrix of corrupted priors,
\begin{equation}\label{eq:prior-covariance-outliers}
	\Var  \doteq\E{\priorallout\priorallout^\top}= I_N + dV, \quad 
	V \doteq \left[\begin{array}{ c | c }
		0 & 0 \\ 
		\hline
		0 & I_M
	\end{array}\right],
\end{equation}
where w.l.o.g. we label malicious \nodes as $ \malSet = \{N-M+1,\dots,N\} $,
and $ V $ is the covariance matrix of noise vector $ v $. 

\begin{lemma}\label{lem:opt-lambda-less-1}
	The optimal parameter
	satisfies $ \lam^* < 1 $.
\end{lemma}
\begin{proof}
	Let us define the following matrices: $ \selreg \in\Real{R\times N} $ maps $ x $ to $ \xreg $,
	and $ \Creg\doteq\frac{1}{R}\one_R\one_R^\top $.
	According to~\eqref{eq:prior-covariance-outliers}, we set $ \selreg = \left[ \, I_R \, \vline \, 0 \, \right] $.
	Then, the error~\eqref{eq:cons-error-regular} can be written as
	\begin{equation}\label{eq:cons-error-regular-all-priors}
		\ereg = 
		\tr{\Var E^\top E}, \quad E \doteq \Emat
	\end{equation}
	and its derivative with respect to $ \lam $ is (up to \revision{multiplicative} constants)
	\begin{equation}\label{eq:cons-error-regular-derivative}
		\dfrac{d\ereg}{d\lam} = \dfrac{1}{\lam}\tr{\Var L^\top\left(I-W^\top L^\top\right)\selreg^\top E},
	\end{equation}
	which at $ \lam = 1 $ takes value
	\begin{equation}\label{eq:cons-error-regular-derivative-at-1}
		\dfrac{d\ereg}{d\lam}\Big|_{\lam=1} = \tr{\Var \left(I-W^\top\right)\selreg^\top \left(\selreg-\Creg\selreg\right)}.
	\end{equation}
	Straightforward computations show that 
	\begin{equation}\label{eq:cons-error-regular-derivative-at-1-arg}
		\Var \left(I-W^\top\right)\selreg^\top \left(\selreg-\Creg\selreg\right) = \left[\begin{array}{ c | c }
		A & 0 \\ 
		\hline
		\star & 0
		\end{array}\right],
	\end{equation}
	where the $ i $th diagonal element of $ A\in\Real{R\times R} $, associated with regular \node $ i\in\regSet $, is
	\begin{equation}\label{eq:cons-error-regular-derivative-at-1-trace-elem}
		a_i = 1 - \dfrac{1}{R}\sum_{m\in\malSet}W_{im} \ge 1 - \dfrac{1}{R} > 0.
	\end{equation}
	Being $ a_i > 0 \, \forall i\in\regSet$,
	the derivative~\eqref{eq:cons-error-regular-derivative-at-1} is positive,
	hence the consensus error~\eqref{eq:cons-error-regular} is strictly increasing
	in a left neighborhood of $ 1 $.
	In virtue of continuity of~\eqref{eq:cons-error-regular-derivative} for $ \lam > 0 $,
	all points of minimum of $ \ereg $ satisfy $ \lam^* < 1 $. 
\end{proof}
Interestingly, the error derivative at $ \lam=1 $ does not depend on the outlier noises,
but only on the (weighted) topology.
\begin{lemma}\label{lem:opt-lam-greater-0}
	The optimal parameter satisfies $ \lam^* > 0 $.
\end{lemma}
\begin{proof}
	In virtue of continuity of the trace operator,
	we can compute the limit of the error derivative~\eqref{eq:cons-error-regular-derivative} as
	\begin{equation}\label{eq:cons-error-regular-derivative-limit}
		\begin{aligned}
			\lim_{\lam\rightarrow0^+}\dfrac{d\ereg}{d\lam} &= \tr{\Var\lim_{\lam\rightarrow0^+}\dfrac{dL}{d\lam}^\top\selreg^\top \lim_{\lam\rightarrow0^+}E}\\
														   &= \tr{\Var\Gamma^\top\selreg^\top\left(\selreg\overline{W}-\Creg\selreg\right)}\\
														   &= \tr{\Var\Gamma^\top\selreg^\top\big[\!-\Creg \,| \,0\big]}\\
														   &= \tr{-\Gamma_{1}^\top\Creg} + \tr{\Gamma_{2}^\top C_{RM}},
		\end{aligned}
	\end{equation}
	where we have used the definitions
	\begin{equation}\label{eq:def-gamma-Wbar}
		\Gamma \doteq \lim_{\lam\rightarrow0^+}\dfrac{dL}{d\lam}, \quad \overline{W} \doteq \lim_{\lam\rightarrow0^+}L, \quad C_{RM} \doteq \frac{\one_R\one_M^\top}{M},
	\end{equation}
	and block partitions (cf.~\cref{ass:mal-node-dynamics} for the value of $ \overline{W} $)
	\begin{equation}\label{eq:Wbar}
		\Gamma = \left[\begin{array}{ c | c }
			\Gamma_{1} & \Gamma_{2} \\ 
			\hline
			0 & 0
		\end{array}\right], \ \overline{W} = \left[\begin{array}{ c | c }
			0 & C_{RM} \\ 
			\hline
			0 & I_M
		\end{array}\right].
	\end{equation}
	Matrix $ \Gamma $ can be computed exactly from the spectral decomposition of $ W $
	(\revision{see Appendix for detailed derivation}).
	In particular, its elements are finite, 
	$ \Gamma_{1} $ is positive and $ \Gamma_{2} $ is negative.
	Hence,~\eqref{eq:cons-error-regular-derivative-limit} is negative
	and $ \ereg $ is strictly decreasing in a right neighborhood of $ \lam = 0 $.
	In virtue of continuity of~\eqref{eq:cons-error-regular-derivative}, we conclude $ \lam^* > 0 $.
\end{proof}

\begin{rem}[Optimal $ \lam $ without outliers]
	\cref{lem:opt-lam-greater-0} implies that $ \lam^* $ is always strictly positive, 
	even with zero outlier noise,
	which may seem counterintuitive.
	This is because 
	the malicious \nodes misbehave
	regardless of their noise level. 
\end{rem}

Lemmas~\ref{lem:opt-lambda-less-1}--\ref{lem:opt-lam-greater-0} are summarized in the following proposition.
\begin{prop}\label{prop:opt-lam-nontrivial}
	The optimal parameter satisfies $ \lam^* \in  (0,1) $.
\end{prop}

\subsubsection{Consensus Error \emph{{vs.}} Noise}\label{sec:opt-lam-vs-d}
We now study how performance is affected by malicious \node
noise variance $ d $.
	
We first prove an intuitive result,
\ie larger outlier noise induces larger consensus error.
\begin{prop}\label{prop:error-increases-with-d}
	The error $ \ereg $ is strictly increasing with $ d $.
\end{prop}
\begin{proof}
	From,~\eqref{eq:cons-error-regular-all-priors} the partial derivative of $ \ereg $ is
	\begin{equation}\label{eq:cons-error-regular-partial-d}
		\dfrac{\partial\ereg}{\partial d} = \tr{VE^\top E}.
	\end{equation}
	Defining the following block partition of $ L $,
	\begin{equation}\label{eq:block-matrix-L}
		L = \left[\begin{array}{ c | c }
			L_{11} & L_{12} \\ 
			\hline
			0 & I_M
		\end{array}\right], \quad L_{11}\in\Real{R\times R}, 
	\end{equation}
	it follows
	\begin{equation}\label{eq:cons-error-regular-partial-d-matrix}
		E^\top E = \left[\begin{array}{ c | c }
		(L_{11} - \Creg)^2 & (L_{11} - \Creg)L_{12} \\ 
		\hline
		L_{12}^\top(L_{11} - \Creg) & L_{12}^\top L_{12}
		\end{array}\right],
	\end{equation}
	and hence $ \tr{VE^\top E} = \tr{L_{12}^\top L_{12}} > 0 $ for $ \lam < 1 $.
	It follows that $ \ereg $ is strictly increasing with $ d $  for all $ \lam < 1 $.
\end{proof}
Intuitively, the more the nominal consensus behavior is disrupted by external attacks,
the more it is convenient for regular \nodes to stick to their own priors
rather than risking to collaborate with malicious \nodes.
Formally speaking, this would require the optimal parameter $ \lam^* $
to increase with the noise level $ d $ and the number of malicious \nodes $ M $.
While the latter behavior cannot be assessed analytically because of the 
discrete nature of \nodes,
the former behavior could be proven
if $ \lam^* $ was unique.
Such a claim is hard to prove because of the structure of the cost function.
In particular, studying the second derivative of $ \ereg $
is complicated by the fact that the trace argument is not positive semidefinite,
and similarly uniqueness of the root of~\eqref{eq:cons-error-regular-derivative}
cannot be assessed in general.
However, the next result contributes towards this intuition,
which is confirmed numerically in~\autoref{sec:numerical-tests-FJ-err}.

\begin{prop}\label{prop:opt-lam-vs-d}
	The critical points of $ \ereg $ are strictly increasing with $ d $.
\end{prop}
\begin{proof}
	We start by computing partial derivatives of the error,
	first with respect to $ \lam $ (cf.~\eqref{eq:cons-error-regular-derivative}) and then with respect to $ d $,
	\begin{equation}\label{eq:cons-error-derivative-lam-d}
		\dfrac{\partial\ereg}{\partial d\partial\lam} = \dfrac{1}{\lam}\tr{LVL^\top\left(I-W^\top L^\top\right)\selreg^\top \selreg},
	\end{equation}
	where we use $ \selreg V = 0 $.
	Defining the block partition
	\begin{equation}\label{eq:block-matrix-W}
		W = \left[\begin{array}{ c | c }
		W_{11} & W_{12} \\ 
		\hline
		0 & I_M
		\end{array}\right], \quad W_{11}\in\Real{R\times R}, 
	\end{equation}
	it follows
	\begin{gather}
		LV = \left[\begin{array}{ c | c }
		0 & L_{12} \\ 
		\hline
		0 & I_M
		\end{array}\right] \label{eq:other-block-matrices}\\
		M \doteq I-W^\top L^\top = \left[\begin{array}{ c | c }
		I_R - W_{11}^\top L_{11}^\top & 0 \\ 
		\hline
		-W_{12}^\top L_{11}^\top - L_{12}^\top & 0
		\end{array}\right] \label{eq:M-block-matrix} \\
		L^\top M\selreg^\top\selreg = \left[\begin{array}{ c | c }
		\star & 0 \\ 
		\hline
		-L_{12}^\top W_{11}^\top L_{11}^\top - W_{12}^\top L_{11}^\top & 0
		\end{array}\right]
	\end{gather}
	and the argument of the trace in~\eqref{eq:cons-error-derivative-lam-d} is
	\begin{equation}\label{eq:cons-error-derivative-lam-d-argument}
		\left[\begin{array}{ c | c }
		-L_{12}L_{12}^\top W_{11}^\top L_{11}^\top - L_{12}W_{12}^\top L_{11}^\top & \star \\ 
		\hline
		\star & 0
		\end{array}\right]
	\end{equation}
	whose upper-left block is a negative matrix for all $ \lam\in(0,1) $,
	and is the zero matrix for $ \lam = 1 $.
	Hence, the error derivative with respect to $ \lam $~\eqref{eq:cons-error-regular-derivative}
	is strictly decreasing with $ d $ for all $ \lam\in(0,1) $.
	In virtue of continuity of~\eqref{eq:cons-error-regular-derivative} in $ \lam $,
	we conclude that the critical points of $ \ereg $ are strictly increasing with $ d $.
\end{proof}

\cref{prop:opt-lam-vs-d} implies that the points of local minimum
increase strictly monotonically with $ d $.
A direct consequence of this fact is that,
if there is a unique critical point
(and therefore a unique point of minimum) for one value of $ d $,
then such a point is always unique
and it is strictly increasing with $ d $.
In words, this entails that higher noise intensity
forces the regular \nodes to progressively reduce collaboration with others,
trusting more their own prior instead.
The next corollary refines this result
by describing the limit behavior of critical points
as the noise intensity grows unbounded.

\begin{cor}
	For any critical point $ \lam_{\text{CR}} $ of $ \ereg $ it holds $ \lim_{d\rightarrow+\infty}\lam_{\text{CR}} = 1 $.
\end{cor}
\begin{proof}
	We first expand~\eqref{eq:cons-error-regular-derivative} to highlight dependence on $ d $:
	\begin{align}\label{eq:cons-error-regular-derivative-d}
		\begin{split}
			\dfrac{d\ereg}{d\lam} = \dfrac{d}{\lam}\tr{-L_{12}\left(L_{12}^\top W_{11}^\top + W_{12}^\top \right) L_{11}^\top} + k(\lam).
		\end{split}
	\end{align}
	It follows that there always exists $ d > 0 $ such that
	the error derivative is negative, for any $ \lam < 1 $.
	In fact, given $ \lam $, the minimal such value of $ d $ is computed as
	\begin{equation}\label{eq:cons-error-regular-derivative-negative}
		d > \lam k(\lam)\tr{L_{12}\left(L_{12}^\top W_{11}^\top L_{11}^\top + W_{12}^\top L_{11}^\top\right)}\inv > 0.
	\end{equation}
	The claim follows by combining this with~\cref{prop:opt-lam-vs-d}.
\end{proof}


\section{Numerical Experiments}\label{sec:numerical-tests-FJ-err}

\begin{figure}
	\centering
	\begin{subfigure}{0.5\linewidth}
		\centering
		\includegraphics[width=.99\linewidth]{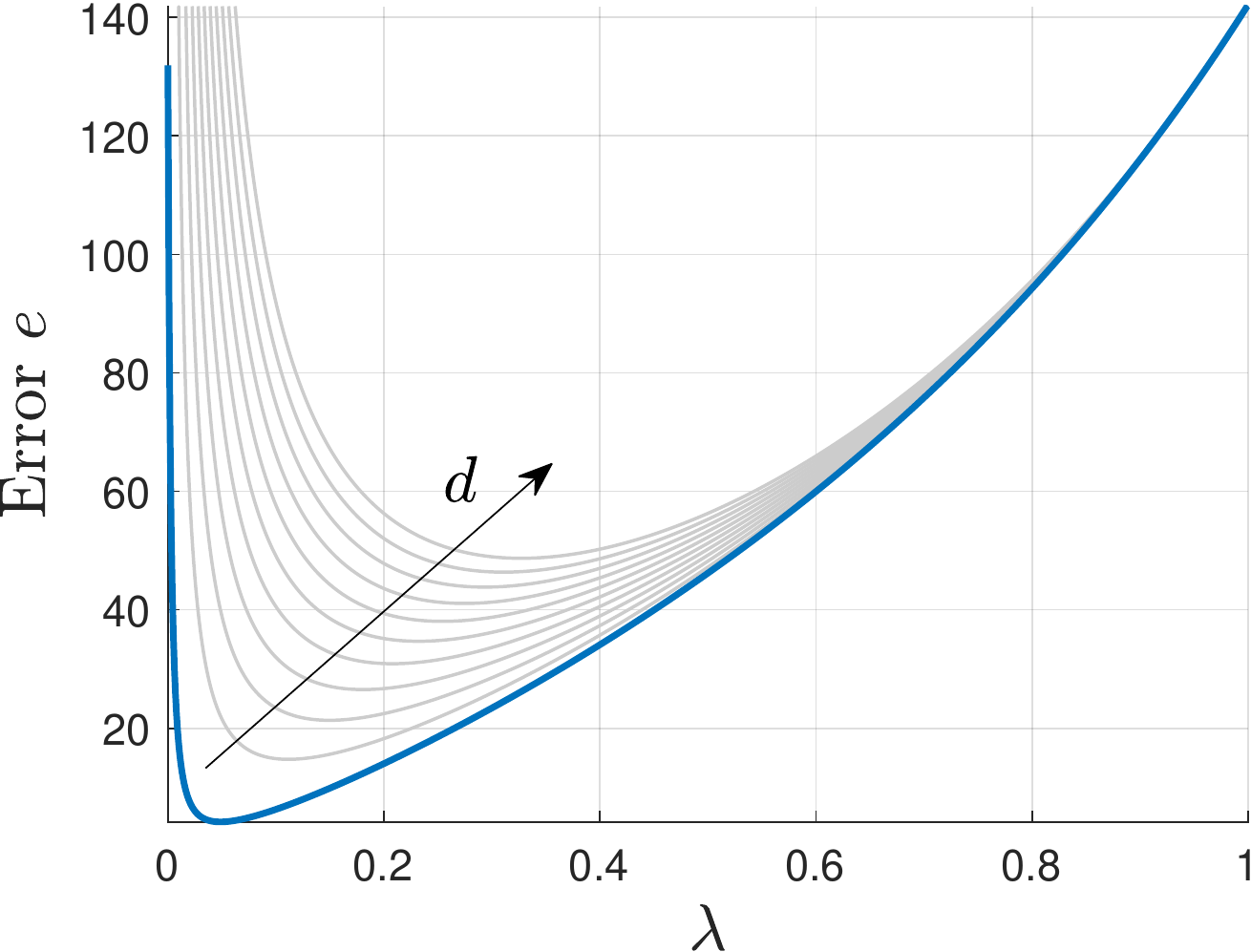}
		\caption{Average consensus error~\eqref{eq:cons-error-regular}}
		\label{fig:err_FJ_reg(3)_d10-100_1mal_curve}
	\end{subfigure}%
	\begin{subfigure}{0.5\linewidth}
		\centering
		\includegraphics[width=.99\linewidth]{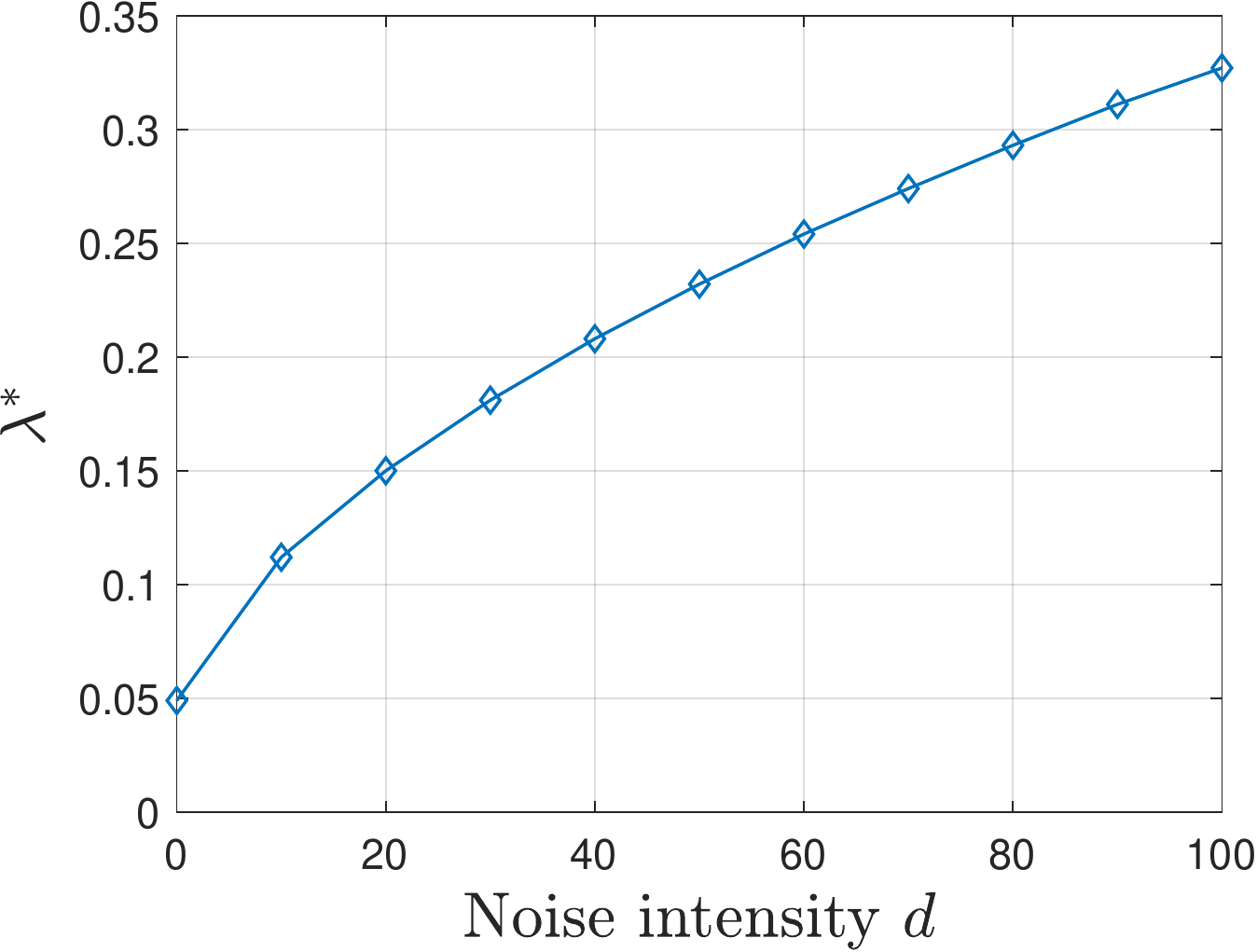}
		\caption{Optimal $ \lam $ as a function of $ d $.}
		\label{fig:err_FJ_reg(3)_d10-100_1mal_opt_lam}
	\end{subfigure}
	\caption{Consensus error of FJ dynamics with $ 3 $-regular graph, $ N = 100 $, one malicious \node, and $ d\in\{0,10,\dots,100\} $.
		The arrow in the left box shows how the error curve varies as the outlier noise intensity $ d $ increases.
	}
	\label{fig:err_FJ_reg(3)_d10-100_1mal}
\end{figure}

In this section,
we compute numerically the consensus error function $ \ereg $~\eqref{eq:cons-error-regular}
in a variety of scenarios,
to see how the FJ dynamics performs under different attacks.

In~\autoref{fig:err_FJ_reg(3)_d10-100_1mal},
we can see the error behavior as the noise intensity varies.
In particular, 
we use a $ 3 $-regular communication graph with $ 100 $ \nodes
and uniform weights. 
Further, 
we select one malicious \node with $ d \in \{0,10,\dots,100\} $.
All error curves in~\autoref{fig:err_FJ_reg(3)_d10-100_1mal_curve} exhibit a unique point of minimum $ \lam^* $
and increase monotonically, 
according to~\cref{prop:error-increases-with-d}.
Further, 
$ \lam^* $ grows with $ d $ (\autoref{fig:err_FJ_reg(3)_d10-100_1mal_opt_lam})
according to~\cref{prop:opt-lam-vs-d}.

\begin{figure}
	\centering
	\begin{subfigure}{0.5\linewidth}
		\centering
		\includegraphics[width=.99\linewidth]{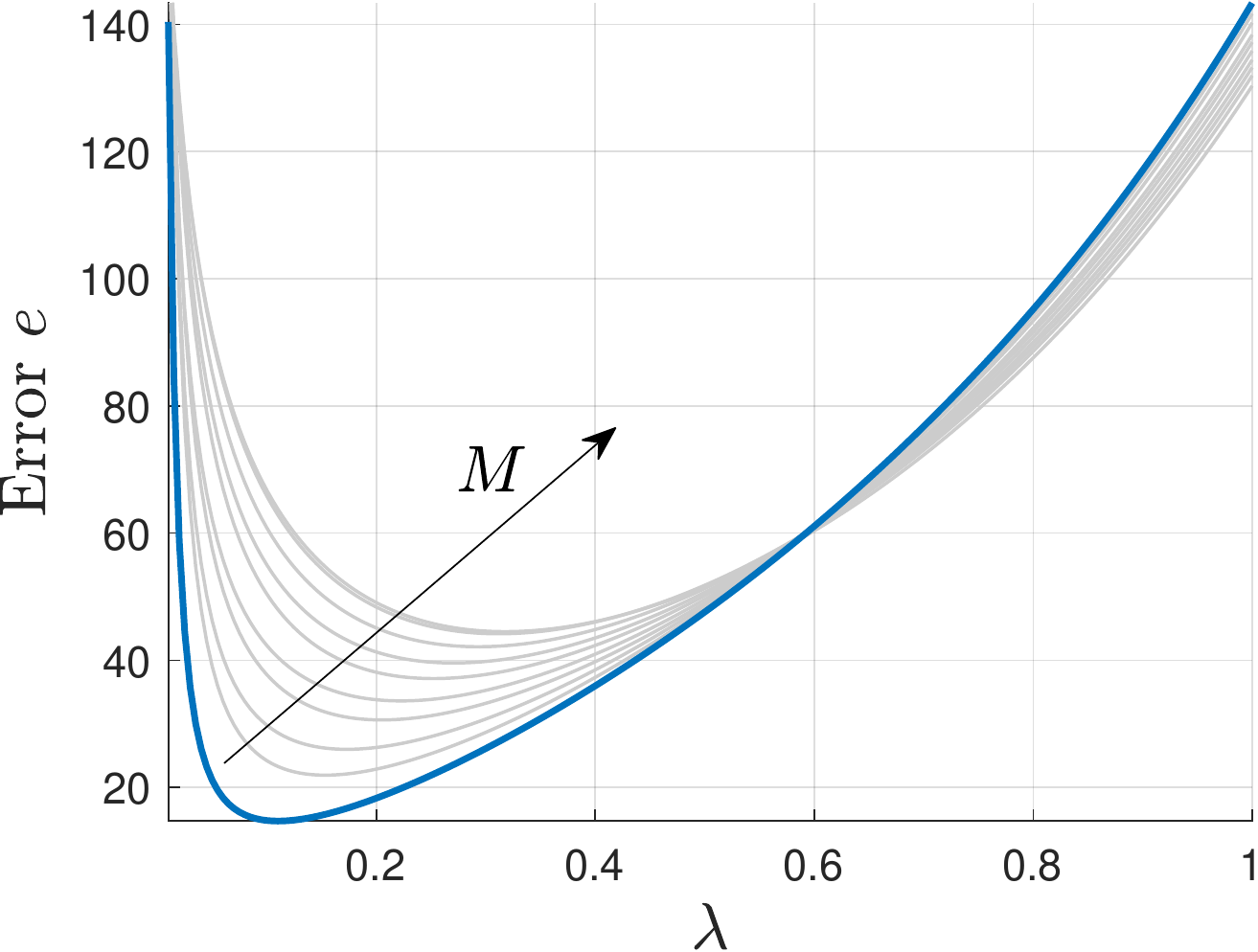}
		\caption{Average consensus error~\eqref{eq:cons-error-regular}}
		\label{fig:err_FJ_reg(3)_d10_mal1-10_curve}
	\end{subfigure}%
	\begin{subfigure}{0.5\linewidth}
		\centering
		\includegraphics[width=.99\linewidth]{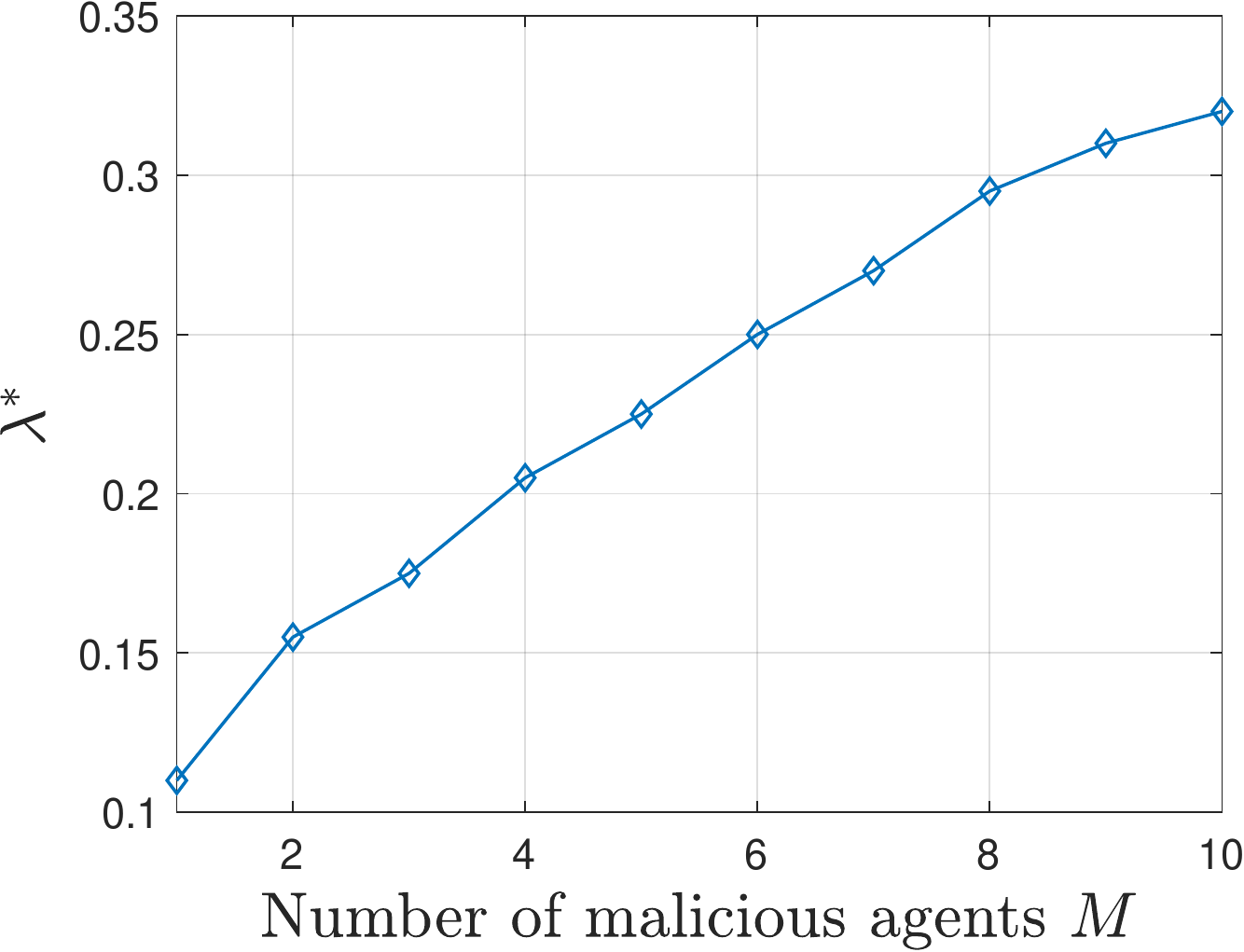}
		\caption{Optimal $ \lam $ as a function of $ M $.}
		\label{fig:err_FJ_reg(3)_d10_mal1-10_opt_lam}
	\end{subfigure}
	\caption{Consensus error of FJ dynamics with $ 3 $-regular graph, $ d = 10 $, $ N = 100 $, and $ M\in\{1,\dots,10\} $ .
		The arrow in the left box shows how the error varies as more regular \nodes turn malicious. 
	}
	\label{fig:err_FJ_reg(3)_d10_mal1-10}
\end{figure}
\begin{figure}
	\centering
	\begin{subfigure}{0.5\linewidth}
		\centering
		\includegraphics[width=.99\linewidth]{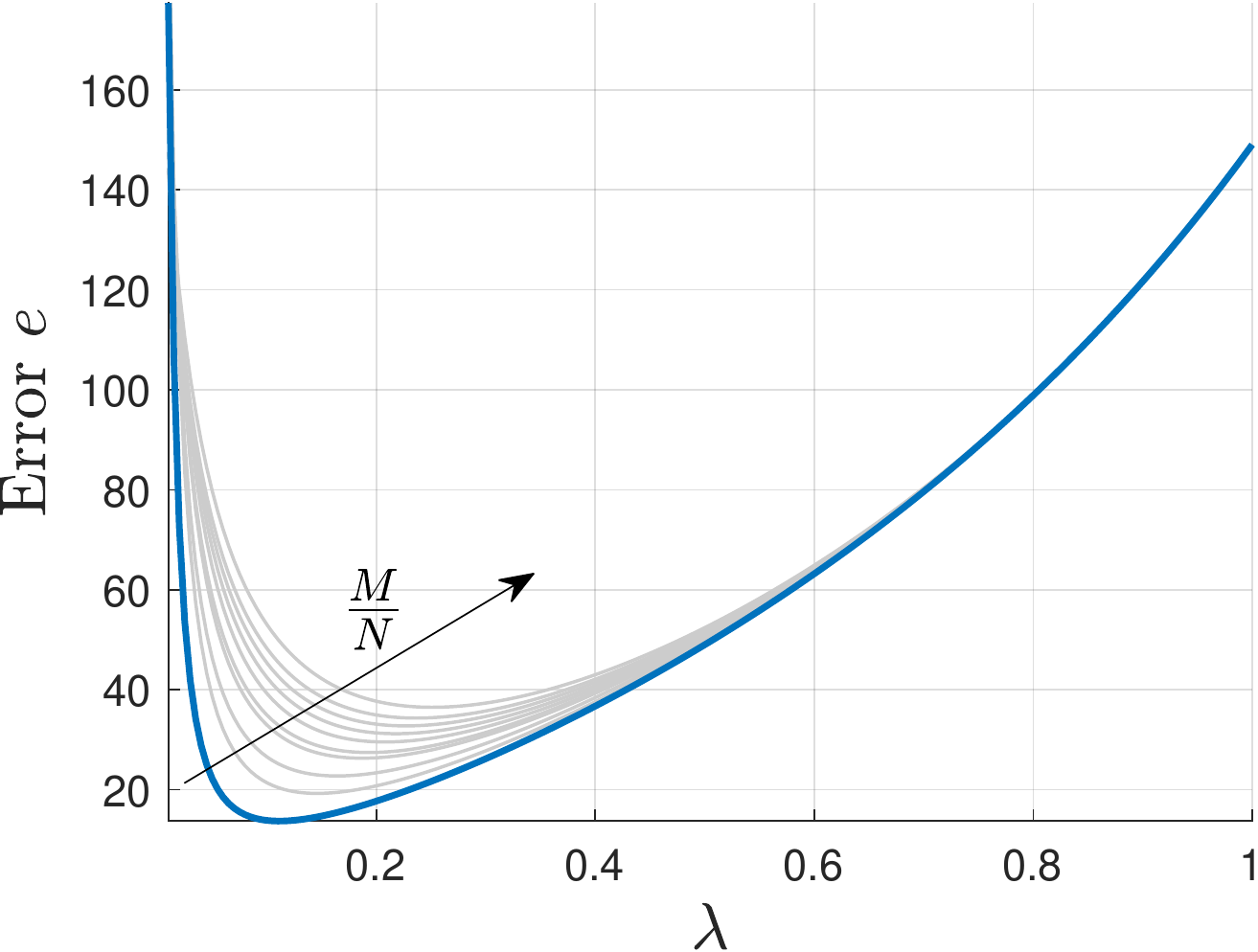}
		\caption{Average consensus error~\eqref{eq:cons-error-regular}}
		\label{fig:err_FJ_reg(3)_d10_mal1-10_reg100_curve}
	\end{subfigure}%
	\begin{subfigure}{0.5\linewidth}
		\centering
		\includegraphics[width=.99\linewidth]{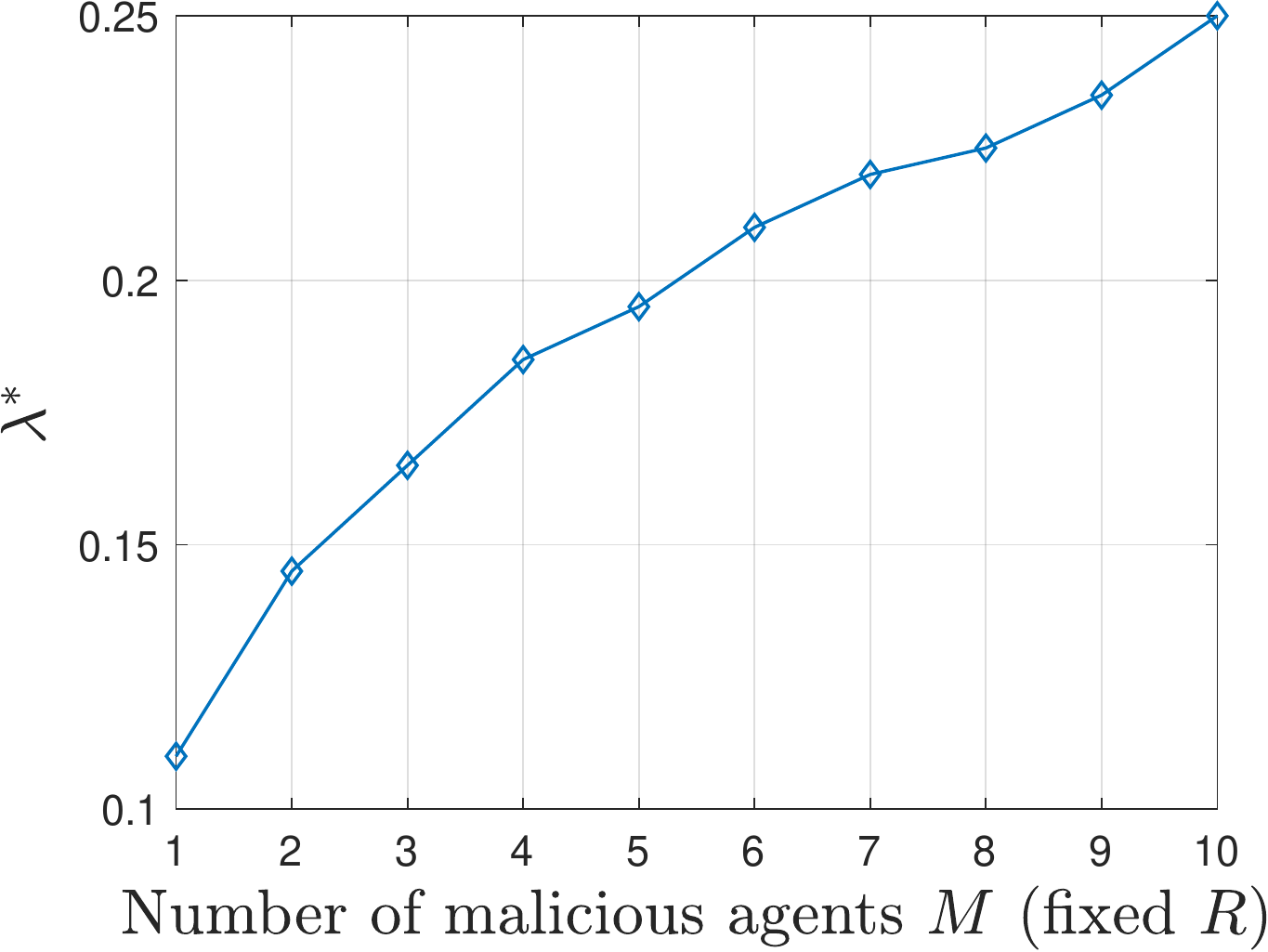}
		\caption{Optimal $ \lam $ as a function of $ M $.}
		\label{fig:err_FJ_reg(3)_d10_mal1-10_reg100_opt_lam}
	\end{subfigure}
	\caption{Consensus error of FJ dynamics with $ 3 $-regular graph, $ d = 10 $, $ R = 100 $, and $ M\in\{1,\dots,10\} $.
		The arrow in the left box shows how the error varies as more malicious \nodes $ M $ are added to the network.
	}
	\label{fig:err_FJ_reg(3)_d10_mal1-10_reg100}
\end{figure}

Analogous behavior is observed when varying the number of malicious \nodes.
We next study what happens by increasing $ M $ while fixing either $ N $ or $ R $.
Figure~\ref{fig:err_FJ_reg(3)_d10_mal1-10} shows
the error curve in a network of $ N = 100 $ \nodes when ten \nodes turn progressively malicious,
while in~\autoref{fig:err_FJ_reg(3)_d10_mal1-10_reg100}
we fix the number of regular \nodes at $ R = 100 $
and add malicious \nodes,
scattering them across the network.
The error curves increase monotonically in~\autoref{fig:err_FJ_reg(3)_d10_mal1-10_reg100_curve},
which is intuitive because regular \nodes face progressively more attacks,
while behave differently in~\autoref{fig:err_FJ_reg(3)_d10_mal1-10_curve},
where the error increases for small values of $ \lam $
but decreases when $ \lam $ grows.
This is due to the lack of a \textquotedblleft normalization"
of the network error,
which decreases with the number of regular \nodes (cf.~\eqref{eq:cons-error-regular-all-priors}).
This emerges for large values of $ \lam $,
when the error is mainly caused by the actual values of \node's priors,
while the error curve increases when $ \lam $ is small.
In both cases, we can see that $ \lam^* $ increases with $ M $ (Figs.~\ref{fig:err_FJ_reg(3)_d10_mal1-10_opt_lam}--\ref{fig:err_FJ_reg(3)_d10_mal1-10_reg100_opt_lam}),
which is because regular \nodes
need to contrast larger and larger amounts of attacks.

\begin{rem}[Value of optimal $ \lam $]
	A remarkable feature of the FJ dynamics,
	which emerges from the above numerical tests,
	is that $ \lam^* $ is relatively small,
	about $ 0.1-0.2 $ for many relevant scenarios.
	In fact, $ \lam^* $ reaches $ 0.3 $
	when, \eg the malicious agent has noise variance $ d = 100 $,
	namely two order of magnitude larger than the variance of priors.
	This translates into the practical advantage that
	adding \textquotedblleft a little" selfishness is sufficient to 
	achieve substantial performance improvement
	compared to standard consensus,
	which may be attractive to get a good level of resilience
	while still letting \node states mix 
	without forcing too conservative updates.
\end{rem}


\section{Comparison with Existing Literature}\label{sec:literature-comparison}

In this section, we compare our algorithm 
with a state-of-the-art resilient optimization strategy,
Weighted Mean Subsequence Reduced (W-MSR), originally proposed for in~\cite{6481629}.
In words, W-MSR ensures that regular agents achieve asymptotic consensus
while maintaining their states within the convex hull of their initial conditions (\textit{resilient consensus}).
Two main features may limit the effectiveness of W-MSR.
The first is that all results rely on the notion of $ r $\textit{-robustness},
which expresses how effective the network is in spreading information.
In particular, most results give sufficient conditions for consensus.
Two practical challenges arise:
on the one hand, the network may be fixed and not enough $ r $-robust,
possibly ruling out MSR-like strategies.
On the other hand, assessing robustness 
in large-scale networks is computationally prohibitive~\cite{7447011}.
Secondly, W-MSR needs to estimate the number of malicious \nodes
to compute the minimal robustness required to succeed. 
This may be problematic if a reliable estimate cannot be provided,
as the updates may be too conservative or misled by adversaries.

Conversely,
a remarkable feature of FJ dynamics is that,
even though it cannot guarantee \emph{average consensus}
(which, in fact, cannot be guaranteed by any resilient algorithm~\cite{7447011}),
it always provides error bounds,
which can be improved by properly tuning $ \lam $.
Further,
while choosing the optimal parameter requires exact knowledge of the adversary,
which is not reasonable,
yet our proposed approach shows good level of robustness to the choice of a specific $ \lam $,
as the plots in~\autoref{sec:numerical-tests-FJ-err} show.
Conversely,
most results in literature do not characterize steady-state behavior of the system
when hypotheses for resilient consensus are not satisfied.
In fact,
they usually either guarantee that \node's states remain in the safety region
(which in practice may be not better than using our approach with $ \lam = 1 $),
or let the \nodes reach consensus but potentially get far away from initial conditions~\cite{8798516}.
Moreover, no extra computation or memory requirements are needed,
as opposed to other algorithms proposed in literature~\cite{8814959}.
This may be relevant to resource-constrained \nodes, 
possibly with real-time requirements.

\begin{figure}
	\centering
	\begin{subfigure}{0.5\linewidth}
		\centering
		\includegraphics[height=.74\linewidth]{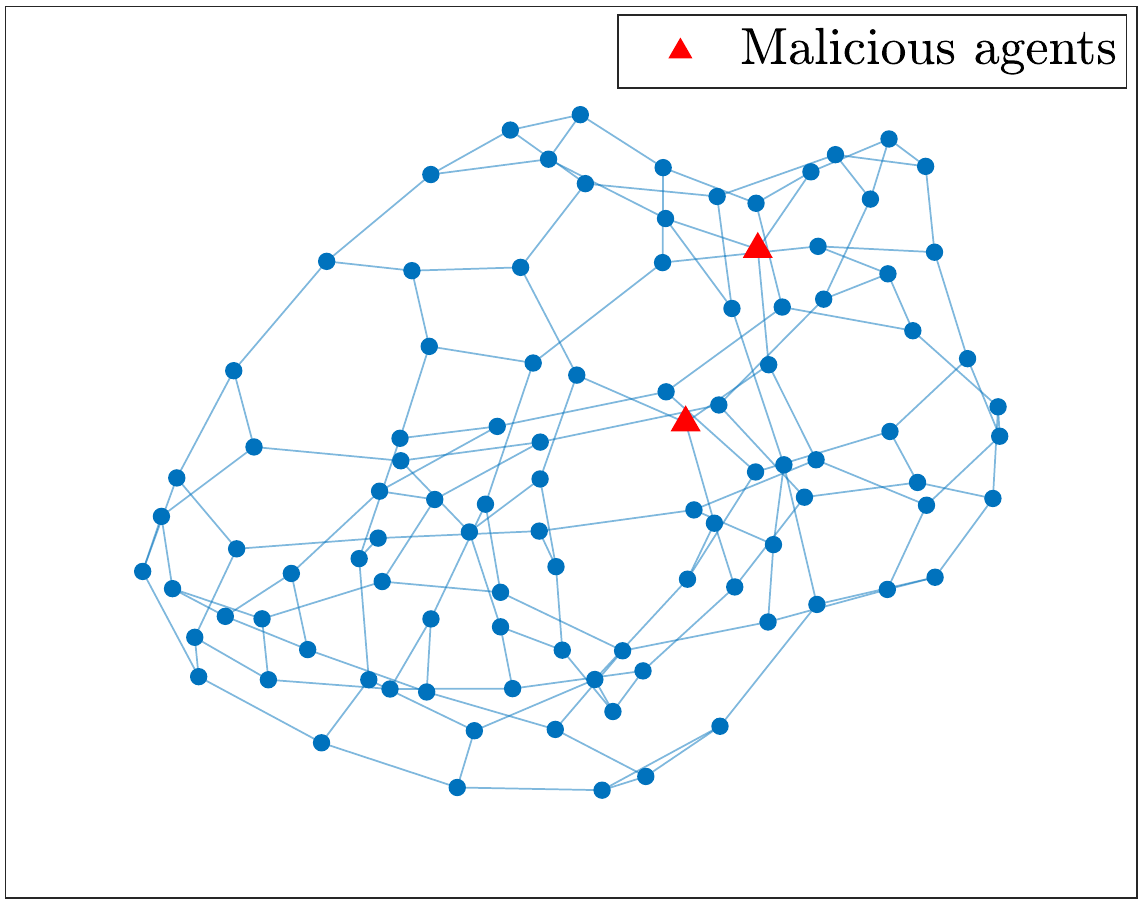}
		\caption{Communication network with malicious \nodes used in simulation.}
		\label{fig:reg(3)_graph}
	\end{subfigure}%
	\begin{subfigure}{0.5\linewidth}
		\centering
		\includegraphics[height=.74\linewidth]{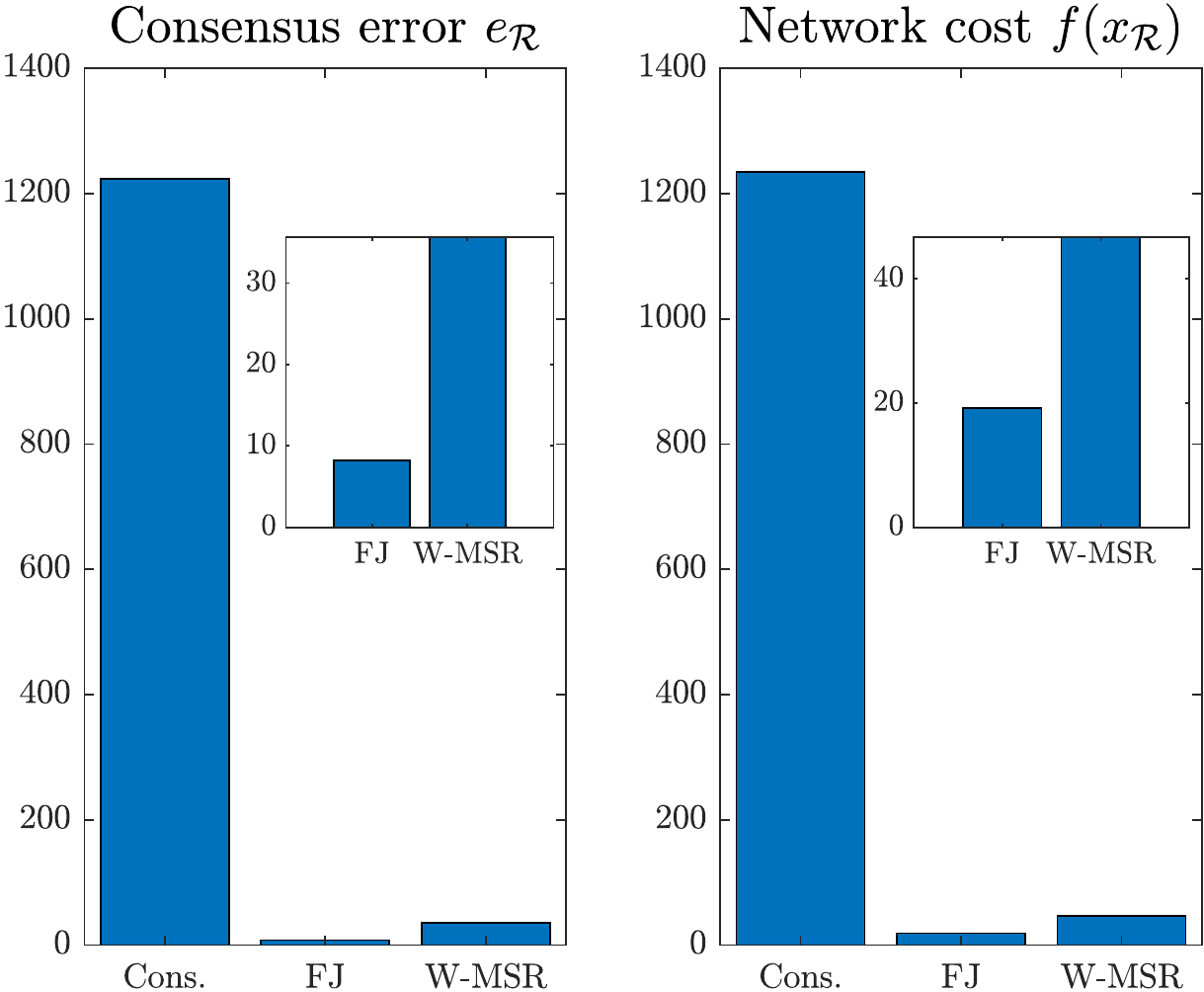}
		\caption{Error~\eqref{eq:cons-error-regular} and cost~\eqref{eq:network-cost-regular} with consensus, FJ dynamics, and W-MSR.}
		\label{fig:reg(3)_net_cost}
	\end{subfigure}
	\caption{Comparison between standard consensus, 
		FJ dynamics (best $ \lam $), 
		and W-MSR with $ 3 $-regular communication graph attacked by $ 2 $ malicious \nodes.
	}
	\label{fig:reg(3)}
\end{figure}

Figure~\ref{fig:reg(3)} shows consensus error and network cost (\autoref{fig:reg(3)_net_cost})
of the two approaches
in a network with $ 100 $ \nodes communicating over a $ 3 $-regular graph (\autoref{fig:reg(3)_graph}),
where the $ 2 $ malicious \nodes are marked as red triangles. 
Note that $ 3 $-regular graphs do not comply with any non-trivial level of $ r $-robustness,
and nothing can be said about W-MSR yielding consensus. 
On the other hand, we compute the theoretical error $ \ereg $ of FJ dynamics
by sampling $ \lam $ in $ (0,1) $,
and use $ \lam = \lam^* $ computed numerically in simulation. 
From~\autoref{fig:reg(3)_net_cost}, we can see that FJ dynamics outperforms W-MSR,
proving to be a competitive approach,
in particular with sparse communication networks.

\section{\titlecap{Conclusion and future research}}\label{sec:future-research}

In the previous sections we have analyzed the theoretical performance
of a model based on FJ dynamics in the presence of malicious \nodes,
bringing insights to algorithmic design.
The proposed approach intends to motivate the general intuition
that competitive approaches may make collaborative tasks
more resilient to external attacks or disturbances,
ideally paving the way to a new viewpoint on strategies for robust and resilient control of networked systems.

This opens several avenues of future research.
First, it is desirable to address an optimal choice of the parameter $ \lam $,
possibly differentiated across \nodes,
in the realistic case where malicious \nodes are unknown.
Further, the algorithm may be improved by letting,
\eg regular \nodes implement smart strategies to detect and isolate adversaries.
Finally, it would be interesting to shift attention to the communication network,
targeting robust design of its topology.
	
	\appendix

\section{\titlecap{decomposition of matrix $ \Gamma $}}\label{app:L-derivative-limit}

We now show how to compute $ \Gamma $~\eqref{eq:Wbar}. 
First, we show how to obtain eigenvalues and eigenvectors of $ \Gamma $ from $ W $.
Second, we show that $ W $ is always diagonalizable,
which implies $ \Gamma $ is also diagonalizable with the same change of basis.

The derivative of $ L $ is 
\begin{equation}\label{eq:L-inv-derivative}
	\dfrac{dL}{d\lam} = \tilde{L} -\lam \tilde{L}\dfrac{d\tilde{L}\inv}{d\lam}\tilde{L} = \tilde{L}-\lam \tilde{L}W\tilde{L},
\end{equation}
where $ \tilde{L}\doteq\left(I-(1-\lam)W\right)\inv $.
Consider the $ i $th eigenvalue of $ W $, denoted as $ \lambda_i $, and its associated eigenvector $ v_i $,
it follows that
$ \tilde{L} $ has $ i $th eigenvalue $ \left(1-(1-\lam)\lambda_i\right)\inv $ with eigenvector $ v_i $.
Hence, straightforward computations yield
\begin{equation}\label{eq:eigendecomposition-derivative-L}
	\dfrac{dL}{d\lam} v_i 
						= \dfrac{1 - \left(1-(1-\lam)\lambda_i\right)\inv\lam\lambda_i}{\left(1-(1-\lam)\lambda_i\right)}v_i\doteq\sigma_i(\lam) v_i.
\end{equation}
In particular, the dominant eigenvector $ v_1 = \one $ (associated with $ \lambda_1 = 1 $)
corresponds to eigenvalue $ \bar{\sigma}_1 = 0 $ for any $ \lam > 0 $.
For $ i > 1 $, by letting $ \lam $ go to zero in~\eqref{eq:eigendecomposition-derivative-L},
one gets
\begin{equation}\label{eq:eigendecomposition-derivative-L-lam-0}
	\bar{\sigma}_i\doteq\lim_{\lam\rightarrow0^+}\sigma_i(\lam) = \left(1-\lambda_i\right)\inv.
\end{equation}
Finally, the eigendecomposition of $ \Gamma $
is obtained by stacking eigenvectors $ \{v_i\}_{i\in\sensSet} $ in $ V $
and eigenvalues $ \{\bar{\sigma}_i\}_{i\in\sensSet} $ in $ D $. 

We now show that $ W $ is always diagonalizable.
Considering the block decomposition~\eqref{eq:block-matrix-W},
and recalling that $ W_{11} $ is symmetric,
it holds
\begin{equation}\label{eq:W-eigendecomposition}
	W = V\left[\begin{array}{ c | c }
		\Lambda_{11} & 0 \\ 
		\hline
		0 & I_M
	\end{array}\right]V\inv,
\end{equation}
where $ W_{11} = V_{11}\Lambda_{11}V_{11} $ and, for any invertible matrix $ V_\malSet $,
\begin{equation}\label{key}
	V = \left[\begin{array}{ c | c }
		V_{11} & (I_R-W_{11})\inv W_{12}V_\malSet\\ 
		\hline
		0 & V_\malSet
	\end{array}\right].
\end{equation}
In particular, $ I_R-W_{11} $ is invertible because the graph is connected~\cite{7577815}.
Hence, eigendecomposition~\eqref{eq:W-eigendecomposition} implies that 
$ \Gamma $ is also diagonalizable through the change of basis $ V $.
	

\end{document}